\documentclass{article}

\usepackage{algorithm}
\usepackage[noend]{algpseudocode}
\usepackage{float}
\usepackage{amsmath}
\usepackage{amssymb}
\usepackage{amsthm}
\usepackage{graphicx}
\usepackage{mathtools}
\usepackage{thmtools, thm-restate}
\usepackage{xcolor}

\DeclarePairedDelimiter\floor{\lfloor}{\rfloor}

\newtheorem{theorem}{Theorem}
\newtheorem{corollary}{Corollary}
\newtheorem{example}{Example}
\newtheorem{strategy}{Strategy}
\newtheorem{lemma}[theorem]{Lemma}

\newtheorem{definition}{Definition}

\begin{document}
\title{Negotiation Strategies for Agents with Ordinal Preferences}
\author {Sefi Erlich \\ \texttt{erlichsefi@gmail.com} \\
	Department of Computer Science, Ariel University, Israel
	\and Noam Hazon \\ \texttt{noamh@ariel.ac.il} \\
	Department of Computer Science, Ariel University, Israel
	\and Sarit Kraus \\ \texttt{sarit@cs.biu.ac.il} \\
	Department of Computer Science, Bar-Ilan University, Israel}

\maketitle
	\begin{abstract}
	Negotiation is a very common interaction between automated agents. Many common negotiation protocols work with cardinal utilities, even though ordinal preferences, which only rank the outcomes, are easier to elicit from humans. In this work we concentrate on negotiation with ordinal preferences over a finite set of outcomes. We study an intuitive protocol for bilateral negotiation, where the two parties make offers alternately. 
	We analyze the negotiation protocol under different settings. First, we assume that each party has full information about the other party's preference order. 
	We provide elegant strategies that specify a sub-game perfect equilibrium for the agents. We further show how the studied negotiation protocol almost completely implements a known bargaining rule. Finally, we analyze the no information setting. We study several solution concepts that are distribution-free, and analyze both the case where neither party knows the preference order of the other party, and the case where only one party is uninformed.
	\end{abstract}
	\section{Introduction}
	%
	Negotiation is a dialogue between two or more parties over one or more issues, where each party has some preferences regarding the discussed issues, and the negotiation process aims to reach an agreement that would be beneficial to the parties. 
	The basic automated negotiation protocol, which consists of two parties that alternate offers, was introduced by Rubinstein~\cite{rubinstein1982perfect}. Since then a lot of work has been done to develop other types of protocols, and to extend the basic bilateral negotiation protocol~\cite{fatima2014principles}. 
	
	Many common negotiation protocols work with cardinal utilities, i.e., with utility functions that give different outcomes a specific numerical value, according to the agents' preferences. This representation requires the agents to specify the magnitude of how much they prefer one outcome over another. However, this specification is not always readily available. Moreover, in many cases the agents need to act on behalf of humans, and the use of cardinal utilities for representing human preferences has been widely criticized on the grounds of cognitive complexity, difficulty of elicitation, and other factors (e.g.,~\cite{ali2012ordinal}). On the other hand, ordinal preferences only rank the outcomes, so they reduce cognitive burden and are easier to elicit.
	
	Indeed, there are some negotiation protocols that work with ordinal preferences. However, these protocols only start with the ordinal preferences, and they then convert them to a cardinal utility according to some assumptions~\cite{nash1950bargaining,shapley1969utility}. Moreover, the traditional assumption in the negotiation theory is that there is a continuum of feasible outcomes. However, many real-life situations involve a finite number of outcomes, such as two managers choosing from among a few job candidates, or a couple choosing from among a few apartments. Even in negotiation over monetary payoffs, the number of outcomes is bounded by the indivisibility of the smallest monetary unit.
	
	In this paper we study negotiation with ordinal preferences over a finite set of outcomes, without converting the ordinal preferences to a cardinal utility
	\footnote{This is also the assumption in most of the voting literature~\cite{handbook}, but see~\cite{boutilier2015optimal} that analyzes social choice functions after extracting cardinal utilities from ordinal preferences.}. 
	We analyze an intuitive protocol for bilateral negotiation that was introduced by Anbarci~\cite{anbarci93noncooperative}, where the two parties make alternating offers. Each offer is a possible outcome, and we allow the parties to make any offer they would like, in any order. The only restriction is that no offer can be made twice, thus if there are $m$ possible outcomes the negotiation will last at most $m$ rounds. 
	
	We analyze the negotiation protocol in different settings. First, we assume that each party has full information about the other party's preference order, and she will thus take this into account and act strategically when deciding what to offer. 
	We provide elegant strategies that specify a Sub-game Perfect Equilibrium (SPE) for the parties. Specifically, our strategies are easy to implement, and we improve the previous result of~\cite{anbarci2006finite} and find a SPE strategy 
	in linear time instead of quadratic time. 
	
	We note that there are several works that studied bargaining rules with ordinal preferences over a finite set of outcomes, but they are inherently different from non-cooperative negotiation protocols. A bargaining rule is a function that assigns to each negotiation instance a subset of the outcomes, which are considered the result of the negotiation. These rules are useful only in a cooperative environment, or where there is a central authority that can force the parties to offer specific outcomes in a specific order. 
	
	The proof that our proposed strategies specify a SPE provides us with a deep understanding of the negotiation protocol, which enables us to establish a connection to the results of the designed \textit{Rational Compromise} ($RC$) bargaining rule~\cite{Bargainingoveraf}. Surprisingly, the SPE result of the negotiation protocol is always part of the set of results returned by the $RC$ rule, even though the protocol does not force the parties to offer specific outcomes in a specific order as the $RC$ rule does. This connection also enables us to prove that the SPE result of the protocol is monotonic.
	
	We then move on to analyze the \textit{no information} setting. We analyze the case where neither party knows the preference order of the other party nor do they know prior probability distribution over possible orders. In this setting we first show that an ex-post SPE does not exist. We then provide the maxmin strategy and the maxmin value of the game that is imposed by our protocol. We also show that in our setting, surprisingly, any pair of maxmin strategies also specifies a robust-optimization equilibrium.
	Finally, we consider the case where one party has full information 
 while the other party has no information, and show how the informed party is able to use her information so that the negotiation result will be better for her.
	
	The contribution of this work is threefold. First, we introduce elegant strategies that specify a SPE, and provide a substantial analysis for showing that they indeed form a SPE. We also provide an improved algorithm for computing 
	a SPE strategy for the studied negotiation protocol. The second contribution of our work is that we show how the studied negotiation protocol almost completely implements the $RC$ rule. As noted by K{\i}br{\i}s and Sertel~\cite{Bargainingoveraf}, who studied the $RC$ rule, the descriptive relevance of the $RC$ rule for real-life bargaining depends on the existence of non-cooperative games that implement it, and to the best of our knowledge our paper is the first to find such a connection. 
	Finally, we provide an analysis of the negotiation protocol under a no information setting, which has not been considered before. 
	\section{Related Work}
	Negotiation protocols and the strategic interaction they imply have been extensively studied. We refer to the books of \cite{osborne90a}, \cite{kraus2001strategic}, and \cite{fatima2014principles} for extensive coverage of the different approaches. The traditional assumption in negotiation theory is that there is a continuum of feasible outcomes, but many real-life negotiation scenarios violate these assumptions. Indeed, there are several works that consider problems with a finite number of outcomes. For example, see \cite{zlotkin1996mechanism}, \cite{mariotti1998nash}, \cite{nagahisa2002axiomatization} and recently \cite{nunez2015bargaining}. All of these works focus on negotiation when the preferences are represented by a cardinal utility, while we study negotiation with ordinal preferences.
	
	Many other works study negotiation with ordinal preferences over a finite set of outcomes (for example, \cite{zhang2008ordinal}). Sequential procedures, in particular the \textit{fallback bargaining} method, have attracted considerable interest \cite{sprumont1993intermediate,hurwicz1999designing,brams2001fallback,Bargainingoveraf,conley2012ordinal}, since they satisfy some nice theoretical properties. All of these works study bargaining rules that are useful in a cooperative environment. We study a negotiation protocol that is useful in a non-cooperative environment, and show that its SPE almost completely implements the individually rational variant of the \textit{fallback bargaining} method, i.e., the $RC$ rule~\cite{Bargainingoveraf}. We note that the $RC$ bargaining rule is equivalent to \textit{Bucklin} voting with two voters, and thus our result can also be interpreted as a (weak) SPE implementation of the Bucklin rule where there are two voters.
	
	There are few works that study negotiation protocols with ordinal preferences over a finite set of outcomes. De Clippel et al.~\cite{de2014selection} study the problem of selection of arbitrators, and they concentrate on two-step protocols. The most closely related works are Anbarci's papers. In~\cite{anbarci93noncooperative} he introduces the \textit{Voting by Alternating Offers and Vetoes} (VAOV) negotiation protocol, which we study here, and shows the possible SPE results in different scenarios. Implicitly, this work shows that the SPE result is unique and Pareto optimal. In~\cite{anbarci2006finite} he introduces three additional negotiation protocols. Moreover, he sharpens his previous result by exactly identifying the SPE result of the VAOV protocol, and by providing an algorithm that computes a SPE strategy. He also shows that if the outcomes are distributed uniformly over the comprehensive utility possibility set and as the number of outcomes tends to infinity the VAOV protocol converges to 
	the equal area rule \cite{thomson1994cooperative}. We provide a more efficient algorithm that finds an elegant SPE strategy. In addition, we were able to establish the relationship between the VAOV protocol and the $RC$ rule, which works with a \textit{finite} number of outcomes, and we also analyze the no information setting. 
\section{The Negotiation Protocol} 
	We assume that there are two negotiation parties, $p^1$ and $p^2$, negotiating over a set of potential outcomes $O=\{o_{1},...,o_{m}\}$, where $p^1$ is the party that makes the first offer. Each party has a preference order over the potential outcomes that does not permit any ties. Formally, the preferences of a party $p$ are a strict order, $\succ_{p}$, which is a complete, transitive and irreflexive binary relation on $O$. We write $o' \succ_{p} o$ to denote that party $p$ strictly prefers $o'$ to $o$, and $o' \succeq_{p} o$ to denote that $o' \succ_{p} o$ or $o'=o$ (i.e $o'$ is the exact same outcome as $o$).
	Clearly, each party would like to maximize her utility, i.e., that the result of the negotiation will be the outcome that is placed as high as possible according to her preferences. 
	
	We analyze the following negotiation protocol, which is the VAOV protocol of~\cite{anbarci93noncooperative}. The parties make offers alternately. No offer can be made twice, but an agreement must be reached since we assume that any agreement is preferred by both parties over a no-agreement result. We also assume that lotteries are not valid offers, as in most real-life negotiations. Formally, denote by $O_t$ the set of available outcomes at round $t$, and the let $O_1=O$. At round $1$, party $p^1$ offers an outcome $o \in O_1$ to $p^2$. If $p^2$ accepts, the negotiation terminates successfully with $o$ as the result of the negotiation. Otherwise, party $p^2$ offers an outcome $o' \in O_2 = O_1 \setminus \{o\}$. If $p^1$ accepts, the negotiation terminates successfully with $o'$ as the result of the negotiation. Otherwise, $p^1$ offers an outcome $o'' \in O_3 = O_2 \setminus \{o'\}$ to $p^2$, and so on. If no offer was accepted until round $m$ then the last available outcome is accepted in the last round as the result of the negotiation. We denote by $p^i$ the party whose turn it is to make an offer at a given round, and by $p^j$ the other party. That is, $p^i = p^1$ in odd round numbers and $p^i = p^2$ otherwise.
	
	We first provide a general result that is useful with any model of information. Consider the following definition:
	\begin{definition}
		In each round $t$, let $L_t^j$ be the $\floor*{{|O_t|}/{2}}$ lowest ranked outcomes in $\succ_{p^j}$. 
		If $|O_t|$ is odd, then let $L_t^i$ be the $\floor*{{|O_t|}/{2}}$ lowest ranked outcomes in $\succ_{p^i}$.
		If $|O_t|$ is even, then $L_t^i$ is the ${|O_t|}/{2}-1$ lowest ranked outcomes in $\succ_{p^i}$. 
		\label{defn:lowers}
	\end{definition}
	We show that in each round $t$ we can identify a set of outcomes that cannot be the negotiation result if the parties are rational, regardless of the information they have. Intuitively, these are all the outcomes that are in the lower parts of the preference orders of both parties, denoted by $Low_t$. We denote all of the other outcomes by $JG_t$. 
	\begin{definition}
		Given a round $t$, let
		$Low_t=\{o : o \in L_t^i \cup L_t^j\}$, and $JG_t = O_t \setminus Low_t$.
	\end{definition}
	\begin{lemma}
		\label{lemma:not_in_lowers}
		Let $o$ be the result of the negotiation if both parties are rational. Then, $o \notin Low_t$.
	\end{lemma}
	\begin{proof}
		Starting from round $t$ where $|O_t|=m_t$, each party will be able to reject all of the offers that she would receive from the other party, except for the offer she would receive in the last round. Specifically, if $m_t$ is odd, $p^{i}$ and $p^{j}$ can reject at most $\floor*{{m_t}/{2}}$ offers. If $m_t$ is even, $p^i$ can reject at most $\floor*{{m_t}/{2}}-1$ offers (since it is $p^{i}$'s turn to offer) and $p^{j}$ can reject at most $\floor*{{m_t}/{2}}$ offers. 
		That is, each party $p^k$, $k\in\{1,2\}$, can reject at most $|L^k_t|$ offers. Therefore, each party will always be able to guarantee that the result of the negotiation will be an outcome that is placed higher than the $|L^k_t|$ lowest outcomes in her preference order. Therefore, $o \notin Low_t$. 
	\end{proof}
	We now analyze the negotiation protocol under two different models of information: full information and no information. In each case we are interested in finding the best actions that a party should take, given the information that she has. 
\section{Full Information}
In this setting we assume that each party has full information about the other party's preference order, and she will thus take this information into account when calculating her best strategy. 
Therefore, in the full information setting we are interested in finding a SPE.
Since \cite{anbarci93noncooperative} showed that the SPE result is unique, it suffices to find one set of strategies that specify a SPE. 

Recall that SPE is a mapping that maps the histories of players' choices. Note that in our case if an offer was accepted the game is over. Therefore, a history for $p^i$, the party whose turn it is to make an offer at a given round, consists of a sequence of outcomes that were proposed and rejected in the previous rounds. Let $H^i_t= (o^1,o^2,...,o^{t-1})$ be the history for $p^i$ at round $t$, and note that $O_t = O \setminus H^i_t$. Let $o^-_t$ be the least preferred outcome in $O_t$ according to $\succ_{p^j}$. We show that the following strategy specifies the offering strategy in our SPE. 
\begin{strategy}[\textsc{Offering Strategy}]
\label{str:offer}
	Given a history $H^i_t$, 
	\textbf{if} $I_t=L^i_t \cap L^j_t \neq \emptyset$ \textbf{then} offer $o \in I_t$, \textbf{else} offer $o^-_t$.
\end{strategy}

A history for $p^j$, the party whose turn it is to decide whether to accept or reject an offer at a given round, consists of a sequence of outcomes that were proposed and rejected in the previous rounds and an additional outcome $o$ that was offered by $p^i$ in the current round. Let $H^j_t = H^i_t + o =  (o^1,o^2,...,o^{t-1},o)$ be the history for $p^j$ at round $t$. In addition, given a round $t$ and history for $p^i$, $H^i_t$, let $o^i_t$ be an outcome $o \in I_t = L^i_t \cap L^j_t$ if $I_t \neq \emptyset$, and $o^-_t$ otherwise. Given a round $t$ and history for $p^j$, $H^j_t$, let $o_a(H^j_t)$ be the single outcome in $O_m = O \setminus H^i_m$, where $H^i_m = H^j_t + o^i_{t+1} + ... + o^i_{m-1}$. That is, $o_a(H^j_t)$ is the result of the negotiation if both parties reject all of the offers that they get (except for the last offer) from round $t$ and on, but use the offering strategy that is specified by Strategy~\ref{str:offer} from round $t+1$ and on. We show that the following strategy specifies the response strategy in our SPE.
\begin{strategy}[\textsc{Response Strategy}]
	\label{str:resp}
	Given a history $H^j_t$,
	\textbf{if} $o \succeq_{p^j} o_a(H^j_t)$ \textbf{then} accept $o$, \textbf{else} reject $o$.
\end{strategy}
To illustrate the strategies of our SPE, consider the following examples:
\begin{example}
	\label{exm:full_info_no_intersection}
	Suppose that
	\[ \succ_{p^1} = o_6 \succ o_5\succ o_4 \succ o_3 \succ o_2 \succ o_1 \]
	\[ \succ_{p^2} = o_1 \succ o_3\succ o_2 \succ o_6 \succ o_4 \succ o_5. \]
	Following Definition~\ref{defn:lowers}, $L_1^1=\{o_2,o_1\}$ and $L_1^2=\{o_6,o_4,o_5 \}$. Therefore, $I_1=\emptyset$ and according to the offering strategy (Strategy~\ref{str:offer}) $p^1$ would offer $p^2$'s least preferred outcome - $o_5$. Then, according to the response strategy (Strategy~\ref{str:resp}) $p^2$ would reject, since $o_a((o_5)) = o_3 \succ_{p^2} o_5$, as we will show. In round $2$, $L_2^1=\{o_1,o_2\}$ and $L_2^2=\{o_6,o_4 \}$, and thus $p^2$ would offer $o_1$. Now $p^1$ would reject, and offer $o_4$, $p^2$ would reject and offer $o_2$, $p^1$ would reject and offer $o_6$, and in the final round $p^2$ would reject and offer $o_3$ which is accepted as the result of the negotiation since no other outcome is available.
\end{example}
\begin{example}
	\label{exm:full_info_intersection}
	Now suppose that
	\[ \succ_{p^1} = o_6 \succ o_5\succ o_4 \succ o_3 \succ o_2 \succ o_1 \]
	\[ \succ_{p^2} = o_1 \succ o_3\succ o_6 \succ o_2 \succ o_4 \succ o_5. \]
	Following Definition~\ref{defn:lowers}, $L_1^1=\{o_2,o_1\}$ and $L_1^2=\{o_2,o_4,o_5 \}$. Therefore, $I_1=\{o_2 \}$ and according to Strategy~\ref{str:offer} $p^1$ would offer $o_2$. Then, $p^2$ will reject, since $o_a((o_2)) = o_6 \succ_{p^2} o_2$, as we will show. In round $2$, $L_2^1=\{o_3,o_1\}$ and $L_2^2=\{o_4,o_5 \}$, $I_2=\emptyset$, and thus $p^2$ would offer $o_1$. In each subsequent round the parties would offer each other the least preferred outcomes, until the final round where $o_6$ will be accepted as the result of the negotiation. 
\end{example}

In order to prove that our strategies specify a SPE, we need a deeper understating of the offering and response strategies.
%
%
We note that in the offering strategy (Strategy~\ref{str:offer}), $p^i$ offers an outcome from the set $I_t$ if it is not empty. We now show the relation between the set $I_t$ and the set $JG_t$. 
\begin{lemma}
	\label{lemma:JG}
	$|JG_t|=|I_t|+1$ 
\end{lemma}
\begin{proof}
	Suppose that in round $t$, $|O_t|=m_t$ is odd. Then
	\[
	|L_t^i \setminus I_t| = \floor*{\frac{m_t}{2}}-|I_t|=\frac{m_t-1}{2}-|I_t|
	\]
	\[
	|L_t^j \setminus I_t|= \floor*{\frac{m_t}{2}}-|I_t|=\frac{m_t-1}{2}-|I_t|
	\]
	Therefore,
	\[
	|JG_t|=m_t-2 \cdot (\frac{m_t-1}{2}-|I_t|)-|I_t|=|I_t|+1
	\]
	Now suppose that $m_t$ is even. Then
	\[
	|L_t^i \setminus I_t | = \floor*{\frac{m_t}{2}}-1-|I_t|=\frac{m_t}{2}-1-|I_t|
	\]
	\[
	|L_t^j \setminus I_t| = \floor*{\frac{m_t}{2}}-|I_t|=\frac{m_t}{2}-|I_t|
	\]
	Therefore,
	\[
	|JG_t|=m_t-(\frac{m_t}{2}-1-|I_t|)-(\frac{m_t}{2}-|I_t|)-|I_t|=|I_t|+1
	\]
\end{proof}
\noindent Considering Lemma~\ref{lemma:not_in_lowers}, we show a simple corollary. Let $o_{eq}$ be the SPE result. We get:
\begin{corollary}
	\label{cor:not_in_lowers}
	$o_{eq} \notin Low_t$.
\end{corollary}
If we combine the findings from Corollary~\ref{cor:not_in_lowers} and Lemma~\ref{lemma:JG}, we get that if the set $I_t$ is empty, i.e., the intersection between the lower parts of the preference orders of the parties is empty, then the set $JG_t$ contains only one outcome, $o_{eq}$.
\begin{corollary}
	\label{lemma:no_inter_eq}
	If $I_t=\emptyset$ then $JG_t = \{o_{eq}\}$.
\end{corollary}
\begin{proof}
	From Lemma~\ref{lemma:JG}, $|JG_t|=1$. Assume that $o_{eq} \notin JG_t$, then $o_{eq} \in Low_t$, in contradiction to Lemma~\ref{lemma:not_in_lowers}.
\end{proof}
Next, we show how the transition from round $t$ to round $t+1$ affects the number of outcomes in $L^k_{t+1}$, $k\in\{1,2\}$.
\begin{lemma}
	\label{lemma:p_2_loss_outcome}
	Suppose that in round $t$, $p^i$ offered an outcome $o$ and $p^j$ rejected it, then in round $t+1$, $|L_{t+1}^i|=|L_{t}^j|-1$ and
	$|L_{t+1}^j|=|L_{t}^i|$ 
\end{lemma}
\begin{proof}
	Assume $|O_t|=m_t$ is even, then by definition $|L_t^i|=\frac{m_t}{2}-1$ and $|L_t^j|=\frac{m_t}{2}$.
	After that $p^i$ offered the outcome $o$ and $p^j$ rejected it, $m_{t+1}$ is odd, and the roles are switched between $p^i$ and $p^j$. Therefore, $|L_{t+1}^i|=|L_{t+1}^j|=\floor*{\frac{m_{t+1}}{2}}=\floor*{\frac{m_t-1}{2}}=\floor*{\frac{m_t}{2}-\frac{1}{2}}=\frac{m_t}{2}-1$.
	Now assume that $m_t$ is odd, then $|L_t^i|=|L_t^j|=\floor*{\frac{m_t}{2}}=\frac{m_t-1}{2}$.
	After that $p^i$ offered the outcome $o$ and $p^j$ rejected it, $m_{t+1}$ is even, and the roles are switched between $p^i$ and $p^j$. Therefore, $|L_{t+1}^j|=\frac{m_t-1}{2}$ and $|L_{t+1}^i|=\frac{m_t-1}{2}-1$ .
\end{proof}
We note that the number of outcomes in $L^k_{t}$ is important, since we already showed in Corollary~\ref{cor:not_in_lowers} that these are the outcomes that cannot be an equilibrium result. Indeed, it is more important to understand how the transition from round $t$ to round $t+1$ affects which outcomes become part of $L^k_{t+1}$. Obviously, it depends on the offer that was made in round $t$. The following three lemmas analyze this transition, based on the offers that are made according to Strategy~\ref{str:offer}. Specifically, Lemma~\ref{lemma:lower_i_not_change} together with Lemma~\ref{lemma:lower_j_not_change} cover the offering strategy where $I_t = \emptyset$, and Lemma~\ref{lemma:lower_j_not_change} together with Lemma~\ref{lemma:lower_i_add} cover the offering strategy where $I_t \neq \emptyset$.
\begin{lemma}
	\label{lemma:lower_i_not_change}
	In round $t$, if $p^i$ offers $o \notin L^i_t$ and $p^j$ rejects it, then $L_{t+1}^j \leftarrow L_{t}^i$ 
\end{lemma}
\begin{proof}
	According to Lemma \ref{lemma:p_2_loss_outcome}, the sets $L^j_{t+1}$ and $L^i_t$ have the same size. Therefore, if $p^i$ offers $o \notin L^i_t$ and $p^j$ rejects it, we can be assured that $L^j_{t+1} = L^i_t$.
\end{proof}
\begin{lemma}
	\label{lemma:lower_j_not_change}
	In round $t$, if $p^i$ offers $o \in L^j_t$ and $p^j$ rejects it, then $L_{t+1}^i \leftarrow L_{t}^j \setminus \{ o \}$
\end{lemma}
\begin{proof}
	According to Lemma \ref{lemma:p_2_loss_outcome}, the set $L^i_{t+1}$ contains one outcome less than the set $L^j_t$. Therefore, if $p^i$ offers $o \in L^j_t$ and $p^j$ rejects it, $o$ is the only outcome that becomes unavailable in round $t+1$, and we can thus be assured that $L_{t+1}^i = L_{t}^j \setminus \{ o \}$.
\end{proof}
\begin{lemma}
	\label{lemma:lower_i_add}
	In round $t$, if $p^i$ offers $o \in L^i_t$ and $p^j$ rejects it, then $L_{t+1}^j \leftarrow L_{t}^i \setminus \{o\} \cup \{o'\}$ 
\end{lemma}
\begin{proof}
	According to lemma \ref{lemma:p_2_loss_outcome}, the sets $L^j_{t+1}$ and $L^i_t$ have the same size. Therefore, if $p^i$ offers $o \in L^i_t$ and $p^j$ rejects it, $o$ is the only outcome that becomes unavailable in round $t+1$, and thus there must be another outcome $o' \in O_{t+1}$ that becomes part of $L^j_{t+1}$.
\end{proof}
%
%
Note that when $p^j$ follows Strategy~\ref{str:resp} she computes the outcome $o_a(H^j_t)$ to decide whether to accept or reject the offer that she gets from $p^i$. 
By definition,
\begin{lemma}
	\label{lemma:the_same}
	$o_a(H^j_t) = o_a(H^j_t + o^i_{t+1})$.
\end{lemma}
We now show that if, in a given round $t$, $p^i$ follows Strategy~\ref{str:offer}, i.e., $o_a(H^j_t) = o_a(H^i_t+o^i_t)$, then $o_a(H^j_t)$ has some desirable properties. For example, it is Pareto optimal in $O_t$, i.e., $\forall o \in O_t \setminus \{o_a(H^i_t+o^i_t)\}$, $o_a(H^i_t+o^i_t) \succ_{p^i} o$ or $o_a(H^i_t+o^i_t) \succ_{p^j} o$.
For ease of notation, let $o_{at} = o_a(H^i_t + o^i_t)$.
\begin{lemma}
	\label{lemma:whatis_oa}
	Given any history $H^i_t$, $o^i_t \prec_{p^j} o_{at}$, $o_{at} \in JG_t$, and $o_{at}$ is Pareto optimal in $O_t$.
\end{lemma}
\begin{proof}
	We prove by induction on $m$. If $m=2$ and $t=1$, without loss of generality (WLOG) assume that $\succ_{p^2} = o_1 \succ o_2$. Since $I_t = \emptyset$, $o^1_t = o^-_t = o_2$ and $H^1_t = ()$ by definition. Thus, $o_{at} = o_a(H^1_t + o^1_t) = o_1$. In addition, $JG_t=\{o_1\}$, and therefore $o^1_t \prec_{p^2} o_{at}$ and $o_{at} \in JG_1$ as required. Since $O_t \setminus \{o_{at}\} = \{o_2\}$ and $o_2 \prec_{p^2} o_{at}$ then $o_{at}$ is also Pareto optimal in $O_t$.	
	Now, assume that if there are $m-1$ outcomes in round $t+1$, $o^i_{t+1} \prec_{p^j} o_a(H^i_{t+1} + o^i_{t+1})$, $o_a(H^i_{t+1} + o^i_{t+1}) \in JG_{t+1}$ and $o_a(H^i_{t+1} + o^i_{t+1})$ is Pareto optimal in $O_{t+1}$. We show that when there are $m$ outcomes in round $t$, $o^i_t \prec_{p^j} o_a(H^i_t + o^i_t) = o_{at}$, $o_{at} \in JG_t$ and $o_{at}$ is Pareto optimal in $O_t$. According to Lemma~\ref{lemma:the_same}, if $H^j_t = H^i_t + o^i_t$ we get that $o_{at} = o_a(H^i_t + o^i_t) = o_a(H^i_t + o^i_t + o^i_{t+1}) = o_a(H^i_{t+1} + o^i_{t+1})$. Now, if $I_t = \emptyset$ and thus $o^i_t = o^-_t \in L^j_t$ then according to Lemmas~\ref{lemma:lower_i_not_change} and \ref{lemma:lower_j_not_change} $JG_{t+1} = JG_t$. If $I_t \neq \emptyset$ and thus $o^i_t \in I_t$ then according to Lemmas~\ref{lemma:lower_j_not_change} and \ref{lemma:lower_i_add} $JG_{t+1} \subseteq JG_t$. According to the induction assumption, $o_a(H^i_{t+1} + o^i_{t+1}) \in JG_{t+1}$ and thus $o_{at} = o_a(H^i_{t+1} + o^i_{t+1}) \in JG_t$. In addition, $o^i_t \in L^j_t$ by definition and since we showed that $o_{at} \in JG_t$, we get that $o^i_t \prec_{p^j} o_{at}$. Finally, since $o_{at}$ is Pareto optimal in $O_{t+1}$, $O_t = O_{t+1} \cup \{o^i_t\}$, and $o^i_t \prec_{p^j} o_{at}$, we conclude that $o_{at}$ is Pareto optimal in $O_t$. 
\end{proof}
Rephrasing Lemma~\ref{lemma:whatis_oa}, we showed that given any history $H^i_t$, if both parties follow Strategies~\ref{str:offer} and \ref{str:resp} from round $t$ and on, $p^j$ would always reject the offers that she gets from $p^i$ (i.e., $o^i_t, o^i_{t+1}, ..., o^i_{m-1}$), and the negotiation result would be $o_{at}$, which would be accepted in the last round. Moreover, the negotiation result $o_{at}$ is Pareto optimal in $O_t, O_{t+1}, ..., O_m$.

Before we prove that Strategies~\ref{str:offer} and \ref{str:resp} specify a SPE we need to add some definitions. We first define a distance function for each party $p^k$, that given an outcome $o_x \notin L^k_t$ counts the number of outcomes $o \notin L^k_t$ such that $o_x \succeq_{p^k} o$. Intuitively, this is the number of outcomes a party can offer until a round $t'$ where $o_x$ becomes part of $L^k_{t'}$. Formally:
%
\begin{definition}
	$d_{k,x,t}=|\{o \in O_t :o_x \succeq_{p^k} o \land o \notin L_t^k \}|$ where $k \in \{1,2\}$
\end{definition}
We also define the number of offers that are made before reaching a round $t'$ where $I_{t'}=\emptyset$.
\begin{definition}
	Let $\ell_{k,t}$ be the number of offers a party $p^k$ offers according to Strategy~\ref{str:offer} from round $t$ until round $t'$ where $I_{t'}=\emptyset$.
\end{definition}
Recall our previous examples.
In Example \ref{exm:full_info_no_intersection} at round 1, $I_1=\emptyset$ and thus $\ell_{1,1}=\ell_{2,1}=0$. The distance of $o_3$ at round $1$ is $d_{1,3,1}=1$ for party $p^1$ and $d_{2,3,1}=2$ for party $p^2$. 
In Example \ref{exm:full_info_intersection}, $I_1 \neq \emptyset$ but $I_2=\emptyset$ and thus $\ell_{1,1}=1$ and $\ell_{2,1}=0$. The distance of $o_6$ at round $1$ for the parties is $d_{1,6,1}=4$ and $d_{2,6,1}=1$, and the distance of $o_3$ at round $1$ for the parties is $d_{1,3,1}=1$ and $d_{2,3,1}=2$.

We also make the following simple observation, which is true since we use an alternating offers protocol: 
\begin{lemma}
	\label{lemma:p1_offer_more}
	At any round $t$, $\ell_{j,t}\leq \ell_{i,t}$. 
\end{lemma}
%
%
%
%
%
%
%
%
%
Our main theorem is as follows:
\begin{theorem}
	Strategies~\ref{str:offer} and \ref{str:resp} specify a SPE.
\end{theorem}
\begin{proof}
	We prove by induction on $m$. If $m=2$ and $t=1$, WLOG assume that $\succ_{p^2} = o_1 \succ o_2$. Thus, $JG_t=\{o_1\}$, and according to Corollary~\ref{lemma:no_inter_eq}, $o_1$ is the SPE result. Indeed, according to Strategy~\ref{str:offer} $p^1$ will offer $o_2$ in the first round and $p^2$ will reject it according to Strategy~\ref{str:resp} since $o_a((o_2))=o_1\succ_2 o_2$. In the next round $p^i = p^2$ will offer $o_1$, $p^1$ will accept it and then the negotiation will end with $o_1$ as the negotiation result. Clearly there are only two states where $p^2$ has the option to deviate: on the equilibrium path, i.e., where $H^2_t = (o_2)$, and off the equilibrium path, i.e., where $H^2_t = (o_1)$. Where $H^2_t = (o_2)$, $p^2$ has no incentive to deviate from the response strategy and accept the offer of $o_2$ from $p^1$, since $o_1\succ_2 o_2$. Where $H^2_t = (o_1)$, $p^2$ has no incentive to deviate and reject the offer of $o_1$, since then $o_2$ would become the last available outcome and thus the negotiation result, but $o_1\succ_2 o_2$. Similarly, there is only one state where $p^1$ has the option to deviate, i.e., $H^1_t = ()$. In this state $p^1$ has no incentive to deviate from the offering strategy and offer $o_1$, since $p^2$ will accept it (because $o_1\succ_2 o_2$) and $o_1$ is already the SPE result if $p^1$ follows the offering strategy.
	
	Now, assume that if there are $m-1$ outcomes in round $t+1$, our strategies specify a SPE. We show that they specify a SPE when there are $m$ outcomes in round $t$. We first consider the strategy of $p^j$ at round $t$. Note that $o_a(H^j_t) = o_a(H^i_t+o) = o_a(H^i_{t+1})$ by definition, and $o_a(H^i_{t+1})$ is the SPE result of following our strategies from state $H^i_{t+1}$ according to Lemma~\ref{lemma:whatis_oa} combined with the induction assumption. Clearly, if according to the response strategy (Strategy~\ref{str:resp}) $p^j$ should reject the offer $o$, it is because $o_a(H^j_t) \succ_{p^j} o$. Therefore, it is not worthwhile for $p^j$ to deviate and accept $o$ instead of $o_a(H^i_{t+1}) = o_a(H^j_t)$. Similarly, if according to the response strategy $p^j$ should accept an offer $o$, it is because $o \succeq_{p^j} o_a(H^j_t)$. Therefore, it is not worthwhile for $p^j$ to deviate and reject $o$ in order to get as the negotiation result the outcome $o_a(H^i_{t+1}) = o_a(H^j_t)$. Overall, $p^j$ does not have an incentive to deviate in round $t$. According to the induction assumption, Strategies~\ref{str:offer} and \ref{str:resp} specify a SPE when there are $m-1$ outcomes in round $t+1$. Therefore, $p^j$ does not have any incentive to deviate. 
	%
	
	We now concentrate on the strategy of $p^i$ at round $t$, but we first derive some general inequalities.
	Given a history $H^i_t$, suppose that there is an outcome $o_x \in O_t$ such that $o_x \succ_{p^i} o_{at} = o_a(H^i_t+o^i_t)$.
	According to Lemma \ref{lemma:whatis_oa}, since $o_x \succ_{p^i} o_{at}$, $o_{at} \succ_{p^j} o_x$. Suppose that both parties follow strategies~\ref{str:offer} and \ref{str:resp}, and let $t'$ be the round in which $JG_{t'}=\{o_{at}\}$. Then, in round $t$, $\ell_{i,t} < d_{i,at,t}$ and $\ell_{j,t} < d_{j,at,t}$ (otherwise, $o_{at} \notin JG_{t'}$). By definition, $t'=\ell_{i,t}+\ell_{j,t}$. In addition, since $JG_{t'}=\{o_{at}\}$, $o_x$ must be part of $Low_{t''}$ for some $t''<t'$ (otherwise, $o_x \in JG_{t'}$). Since $o_x \succ_{p^i} o_{at}$ and $\ell_{i,t} < d_{i,at,t}$, it must be that $o_x$ is part of $L^j_{t''}$, that is, $d_{j,x,t} \leq \ell_{j,t}$ . In summary:
	\begin{equation}
	\label{eq:not_defact}
	\begin{aligned}
	\ell_{i,t} < d_{i,at,t} < d_{i,x,t} \\ 
	d_{j,x,t} \leq \ell_{j,t} < d_{j,at,t}
	\end{aligned}
	\end{equation}
	
	Now assume that in round $t$ $p^i$ deviates, and the result of the negotiation, if both parties follow our strategies from round $t+1$, is $o_x$. Note that $p^i$ in round $t+1$ is $p^j$ in round $t$, and thus $o_{at} \succ_{p^i} o_x$. Therefore, we use the same arguments as above to get 
	\begin{equation}
	\label{eq:defact}
	\begin{aligned}
	\ell_{i,t+1} < d_{i,x,t+1} < d_{i,at,t+1} \\
	d_{j,at,t+1} \leq \ell_{j,t+1} < d_{j,x,t+1} 
	\end{aligned}
	\end{equation}
	%
	%
	
	Now, assume by contradiction that there is an outcome $o_d \notin I_t$ such that if $p^i$ offers $o_d$ the negotiation result will be $o_x$, $o_x \succ_{p^i} o_{at}$. We first analyze the case where $p^j$ rejects the offer of $o_d$, since $o_x \succ_{p^j} o_d$ (otherwise, $p^j$ would have accepted). We examine the change in the distance function for $p^i$ and $p^j$, for outcomes $o_{at}$ and $o_x$, from round $t$ to round $t+1$. According to Lemma~\ref{lemma:p_2_loss_outcome}, $|L_{t+1}^i|+1=|L_{t}^j|$, and since $o_x \succ_{p^j} o_d$ and $o_{at} \succ_{p^j} o_d$, $d_{j,x,t}$ and $d_{j,at,t}$ do not change when moving to round $t+1$.
	Let $c$ be an integer. Then, $d_{i,at,t}+c=d_{j,at,t+1}$, $d_{j,at,t}=d_{i,at,t+1}$, and $d_{j,x,t}=d_{i,x,t+1}$.
	%
	If we combine these $3$ equations with equations~\ref{eq:defact} we get that $\ell_{j,t+1} < d_{j,x,t} < d_{j,at,t}$ and $d_{i,at,t}+c \leq \ell_{i,t+1}$.
	%
	Adding equations \ref{eq:not_defact} we get that $\ell_{i,t} < \ell_{i,t+1}-c$ and $\ell_{i,t+1}< \ell_{j,t}$.
	%
	Adding Lemma~\ref{lemma:p1_offer_more} we can conclude that $\ell_{i,t+1}< \ell_{j,t} \leq \ell_{i,t}< \ell_{i,t+1}-c$. That is, $\ell_{i,t+1} \leq \ell_{i,t+1}-c-2$, thus $c \leq -2$.
	However, the distance function cannot decrease by more than $1$ when moving from round $t$ to $t+1$, thus $c \geq -1$.
	
	We now analyze the case where $p^j$ accepts the offer of $o_d$, since $o_d \succeq_{p^j} o_x$. We examine the change in the distance function for $p^i$ and $p^j$, for outcomes $o_{at}$ and $o_x$, from round $t$ to round $t+1$. Note that since $p^i$ deviates, $o_d \succ_{p^i} o_{at}$. According to Lemma~\ref{lemma:whatis_oa}, $o_{at} \succ_{p^j} o_d$. According to Lemma~\ref{lemma:p_2_loss_outcome}, $|L_{t+1}^i|+1=|L_{t}^j|$, and since $o_{at} \succ_{p^j} o_d$, $d_{j,at,t}$ does not change when moving to round $t+1$. However, since $o_d \succ_{p^j} o_x$, $d_{j,x,t}$ increases by one when moving to round $t+1$. Let $c$ be an integer. Then,
	$d_{i,at,t}+c=d_{j,at,t+1}$, $d_{j,at,t}=d_{i,at,t+1}$, and $d_{j,x,t}+1=d_{i,x,t+1}$.
	%
	If we combine these $3$ equations with equations~\ref{eq:defact} we get that $\ell_{i,t+1} < d_{j,x,t}+1 < d_{j,at,t}$ and $d_{i,at,t}+c \leq \ell_{j,t+1}$.
	%
	Adding equation $\ref{eq:not_defact}$ we get that $\ell_{i,t} < \ell_{j,t+1}-c$ and $\ell_{i,t+1}-1< \ell_{j,t}$.
	Adding Lemma~\ref{lemma:p1_offer_more} we can conclude that $\ell_{i,t+1}-1< \ell_{j,t} \leq \ell_{i,t}< \ell_{j,t+1}-c$. That is, $\ell_{j,t+1}-1 \leq \ell_{j,t+1}-c-2$, thus $c \leq -1$. However, in order for $d_{i,at,t}$ to decrease by at least one, $o_{at} \succ_{p^i} o_d$, but in our case $o_d \succ_{p^i} o_{at}$.
    
    Overall, we showed that $p^i$ does not have an incentive to deviate in round $t$. According to the induction assumption, Strategies~\ref{str:offer} and \ref{str:resp} specify a SPE when there are $m-1$ outcomes in round $t+1$. Therefore, $p^i$ does not have any incentive to deviate. 
\end{proof}
Finally, note that trivial exploration of the whole game tree in order to derive the SPE would take at least $O(2^m)$ operations, since there can be $m-1$ rounds in which a party $p^i$ can offer any outcome from the available outcomes and the other party $p^j$ can decide either to accept the offer or reject it. 
The complexity of finding a SPE strategy of~\cite{anbarci2006finite} is not explicitly analyzed, but its running time is at least $O(m^2)$. 
Our approach provides elegant strategies that are easy to implement and are (computationally) more efficient: given a state in the game tree (i.e. given any history $H^i_t$ or $H^j_t$), we compute  
a SPE strategy from the current state in time that is linear in $m$. Indeed, in our approach we only need to simulate one branch of the tree (to find $o_a(H^i_t)$ or $o_a(H^j_t)$) and then trace the intersection between $L^i_t$ and $L^j_t$.

\subsection{Properties}
We first note that since we showed that the result of following Strategies~\ref{str:offer} and \ref{str:resp} is Pareto optimal, we proved that they specify a SPE, and the SPE result is unique, we can infer that the SPE result of the protocol is Pareto optimal. We now move to analyze the relationship between the SPE result of the protocol and the results of the designed \textit{Rational Compromise} ($RC$) bargaining rule~\cite{Bargainingoveraf}. The $RC$ rule is a private case of the \textit{Unanimity Compromise} rule, where any agreement is preferred by both parties over a no-agreement result, as we assume. 
With our notations, the $RC$ rule can be rephrased as the set $RC=\{o_x | \max_{o_x \in O } \text{min}_{k \in \{1,2\}} (d_{k,x,1}+|L^k_1|-1)\}$. It can also be computed by the following steps:
\begin{enumerate}
	\item Let $v=1$
	\item For each $k \in \{1,2\}$, let $B^k_v=\{$the $v$ most preferred outcomes in $\succ_{p^k} \}$.
	\item If $|B^1_v \cap B^2_v| > 0$ then return $B^1_v \cap B^2_v$ as the result.
	\item Else, $v \leftarrow v+1$ and go to line $2$. 
\end{enumerate}
We note that the $RC$ rule may return either one or two outcomes, while our protocol always results with a single outcome. Surprisingly, the SPE result of the negotiation protocol is always part of the set returned by the $RC$ rule. The intuition is that our strategies specify a SPE by making offers and rejecting them until $I_t=\emptyset$. At this stage $JG_t = \{o_{eq}\}$, and by definition the set $JG_t$ is the intersection of the upper parts of the preferences of both parties, which corresponds to the $B^1_v \cap B^2_v$ returned by $RC$.
\begin{theorem}
	\label{thm:in_rc}
	$o_{eq} \in RC$
\end{theorem}
\begin{proof}
	Let $t$ be the round where $I_t=\emptyset$ after both parties follow our strategies. By Corollary~\ref{lemma:no_inter_eq}, $JG_t=\{o_{eq}\}$. Rephrasing the definition of $JG_t$ we get that $JG_t = B^i_{|O_t|-|L^i_t|} \cap B^j_{|O_t| -|L^j_t|}$.
	If $|L^j_t|=|L^i_t|$, then for any $v$ where $v \leq |O_t|-|L^j_t|$ , $B^i_v \cap B^j_v=\{o_{eq}\}$ or $B^i_v \cap B^j_v=\emptyset$. If $|L^j_t|=|L^i_t|+1$, then for any $v$ where the $v \leq |O_t|-|L^j_t|$, $B^i_v \cap B^j_v=\{o_{eq}\}$ or $B^i_v \cap B^j_v=\emptyset$, and for $v=|O_t|-|L^i_t|$ it is possible that $B^i_v \cap B^j_v=\{o_{eq},o_x\}$, for some outcome $o_x$. Overall, $o_{eq} \in RC$.
\end{proof} 
%
%
%
Based on Theorem~\ref{thm:in_rc}, we can derive interesting results regarding the relationship between the $RC$ rule and the SPE result of the negotiation protocol:
\begin{theorem}
\begin{enumerate}
	\item If $RC=\{o\}$ then $o_{eq}=o$.
	\item If $o_{eq}$ is the SPE result let $o_{eq'}$ be the SPE result if $p^1$ and $p^2$ switch their rules (i.e., $p^2$ starts the negotiation). If $o_{eq} \neq o_{eq'}$, then $RC=\{o_{eq},o_{eq'}\}$
	\item If $m$ is odd and $\ell_{1,1}+\ell_{2,1}$ is even or if $m$ is even and $\ell_{1,1}+\ell_{2,1}$ is odd, then $|RC|=1$.
	\item If $|RC|=\{o_x,o_y\}$ and $\ell_{1,1}+\ell_{2,1}$ is odd then $o_{eq}=o_x$ and $o_x \succ_{p^i} o_y$. If $\ell_{1,1}+\ell_{2,1}$ is even then $o_{eq}=o_y$ and $o_y \succ_{p^j} o_x$.
\end{enumerate}
\end{theorem}
\begin{proof}
	\begin{enumerate}
		\item An easy corollary of Theorem~\ref{thm:in_rc}.
		\item An easy corollary of Theorem~\ref{thm:in_rc}.
		\item If $m$ is odd and $\ell_{1,1}+\ell_{2,1}$ is even or if $m$ is even and $\ell_{1,1}+\ell_{2,1}$ is odd, then $m_t$ is odd. Therefore, $|L^1_t| = |L^2_t|$ by definition. Then, by Theorem~\ref{thm:in_rc}, for any $v$ where $v \leq |O_t|-|L^1_t|$ , $B^1_v \cap B^2_v=\{o_{eq}\}$ or $B^1_v \cap B^2_v=\emptyset$. That is, $RC = \{o_{eq}\}$.
		\item $|RC|=2$, thus there exists $v$ such that $B^1_v \cap B^2_v=\{o_x,o_y\}$, and for every $v' < v$, $B^1_{v'} \cap B^2_{v'} = \emptyset$. From Theorem~\ref{thm:in_rc}, $o_{eq} = o_x$ or $o_{eq} = o_y$. Let $t$ be the round such that $I_t = \emptyset$ and $JG_t = \{o_{eq}\}$. That is, $B^1_{|O_t|-|L^1_t|} \cap B^2_{|O_t| -|L^2_t|} = \{o_{eq}\}$. Therefore, $|L^1_t| \neq |L^2_t|$, and thus $m_t$ is even. If $\ell_{1,1}+\ell_{2,1}$ is odd then it is $p^2$'s turn to offer. That is, $|L^2_t| + 1 = |L^1_t|$, and since $o_x \succ_{p^1} o_y$, $o_y \in L^1_t$. Therefore, $o_{eq} = o_x$. Similarly, if $\ell_{1,1}+\ell_{2,1}$ is even then it is $p^1$'s turn to offer. That is, $|L^1_t| + 1 = |L^2_t|$, and since $o_y \succ_{p^2} o_x$, $o_x \in L^2_t$. Therefore, $o_{eq} = o_y$.
	\end{enumerate}
\end{proof}
Finally, we adapt the monotonicity criterion that the $RC$ rule satisfies for our domain, and show that the negotiation protocol is monotonic. 
\begin{definition}
	A negotiation protocol is monotonic if given an instance $(O,\succ_{p^1},\succ_{p^2})$ where the SPE result is $o_{eq}$, then for any instance $(O',\succ_{p^1}',\succ_{p^2}')$ such that:
	\begin{enumerate}
		\item $O \subset O'$.
		\item For any $o_1,o_2 \in O, o_1 \neq o_2$, and for $k \in \{1,2\}$, if $o_1 \succ_{p^k} o_2$ then $o_1 \succ_{p^k}' o_2$.
		\item For any $o \in O' \setminus O$, and for $k \in \{1,2\}$, $o \succ_{p^k}' o_{eq}$.
	\end{enumerate}
	we have that $o'_{eq} \succ_{p^k}' o_{eq}$.
\end{definition}
\begin{theorem}
	The negotiation protocol is monotonic.
\end{theorem}
\begin{proof}
	Given an instance $(O,\succ_{p^1},\succ_{p^2})$, we know from Theorem~\ref{thm:in_rc} that $o_{eq} \in RC$. If we add a set of outcomes $O'\setminus O$ such that for every outcome $o \in O'\setminus O$, $o \succ o_{eq}$ for both parties, then for every outcome $o'$ in the set returned by the $RC$ rule on the modified instance $(O',\succ_{p^1}',\succ_{p^2}')$, $o' \succ o_{eq}$ by both parties.
	Since $o_{eq}' \in RC$ on $(O',\succ_{p^1}',\succ_{p^2}')$, we get that $o_{eq}' \succ o_{eq}$ for both parties, as required.
\end{proof}
%
%
%


%
%
%
\section{No Information}
	We now consider the case of no information, where we assume that neither party knows the preference order of the other party. Moreover, the parties do not even hold any prior probability distribution over each other's possible preference orders. A common solution concept for this case is an ex-post equilibrium, or in our case, an ex-post SPE. Intuitively, this is a strategy profile in which the strategy of each party depends only on her own type, i.e., its preference order, and it is a SPE for every realization of the other party's type (i.e. her private preference order). Formally, let $s_k(\prec)$ be a strategy for player $k \in \{1,2\}$ given a preference order $\prec$, and let $\mathcal{F}([s_1(\prec),s_2(\prec')])$ be the negotiation result if both parties follow their strategies. In the ex-post setting, a strategy for party $k \in \{1,2\}$,  $s_k$, is a best response to $s_{3-k}$ if for every strategy $s_k'$ and for every preference orders $\prec,\prec'$, $\mathcal{F}([s_k(\prec),s_{3-k}(\prec')]) \succeq_k \mathcal{F}([s_k'(\prec),s_{3-k}(\prec')])$. A strategy profile $[s_1,s_2]$ is an ex-post equilibrium if $s_1$ is a best response to $s_2$ and $s_2$ is a best response to $s_1$, and it is an ex-post SPE if it is an ex-post equilibrium in every subgame of the game. We show that ex-post equilibrium, and thus also ex-post SPE, are too strong to exist in our setting. 
	\begin{theorem}
		\label{thm:no_ex_post}
		There are no two strategies that specify an ex-post equilibrium for our protocol.
	\end{theorem}
	\begin{proof}
    Clearly, every ex-post SPE is also a SPE (i.e., in the full information setting), and we can thus use our previous results that characterize the SPE. Assume by contradiction that there are two strategies $s_1,s_2$ for parties $p^1,p^2$, respectively, such that $[s_1,s_2]$ is an ex-post SPE. Let $\prec_1 = o_1 \prec o_2 \prec o_5 \prec o_4 \prec o_3 \prec o_6$, and let $\prec_2 = o_1 \prec o_2 \prec o_6 \prec o_3 \prec o_5 \prec o_4$. Following our strategies we get that the SPE result is $o_4$. Since the SPE result is unique and every ex-post SPE is also a SPE, $\mathcal{F}([s_1(\prec_1),s_2(\prec_2)]) = o_4$. Now consider $\prec_1' = o_4 \prec o_5 \prec o_1 \prec o_2 \prec o_3 \prec o_6$. According to Corollary~\ref{lemma:not_in_lowers}, $\mathcal{F}([s_1(\prec_1'),s_2(\prec_2)]) = o_3$. Note that $o_4 \prec_{p^1} o_3$. Consider the following strategy: $s_1'(\prec) = s_1(\prec_1')$ if $\prec=\prec_1$, and $s_1'(\prec) = s_1(\prec)$ otherwise. That is, $\mathcal{F}([s_1(\prec_1),s_2(\prec_2)]) \prec_{p^1} \mathcal{F}([s_1'(\prec_1),s_2(\prec_2)])$, and thus $s_1$ is not a best response to $s_2$.  
	\end{proof} 
	We note that Theorem~\ref{thm:no_ex_post} also implies that there is no solution in dominant strategies.
	Another approach to uncertainty, which follows a conservative attitude, is that a party $p^k$, $k\in \{1,2\}$, who wants to maximize her utility may want to play a maxmin strategy. That is, since the preference order and the strategy of the other party $p^{3-k}$ are not known, it is sensible to assume that $p^{3-k}$ happens to play a strategy that causes the greatest harm to $p^k$, and to act accordingly. $p^k$ then guarantees the maxmin value of the game for her, which in our case is a set of outcomes such that no other outcome that is ranked lower than all of the outcomes in this set will be accepted as the result of the negotiation, regardless of the preferences of $p^{3-k}$.
	Before we show the maxmin strategy we define the complement sets for the sets $L^k_t$, i.e., the sets of highest ranked outcomes. 
	\begin{definition} 
		In each round $t$, for each party $p^k$, $k \in \{1,2\}$,
		$U^k_t= O_t \setminus L^k_t$.
		\label{defn:upper}
	\end{definition}
	The maxmin strategy, which is composed of offering and response strategies, is defined as follows:
	\begin{strategy}[\textsc{Maxmin Strategy}]
		\label{str:maxmin}
		Given a history $H^i_t$, offer any $o \in U^i_t$.
		Given a history $H^j_t$,
		\textbf{if} $o \in U^j_t$ \textbf{then} accept $o$, \textbf{else} reject $o$.
	\end{strategy}
	We now prove that our strategy specifies a maxmin strategy, and that a party $p^k$ that follows it can guarantee the maxmin value of the game, which is the set $U^k_1$.
	We denote the party that uses Strategy~\ref{str:maxmin} by $p^{max}$ and the other party, which might try to minimize the utility of $p^{max}$, by $p^{min}$. Note that we need to handle both the case where $p^{max}$ starts the negotiation (i.e, $p^{max}=p^1$) and the case where $p^{min}$ starts it (i.e., $p^{min}=p^1$).
	We re-use Lemmas \ref{lemma:lower_i_not_change}, \ref{lemma:lower_j_not_change} and \ref{lemma:lower_i_add}, since they do not depend on the full-information assumption. Furthermore, we add a fourth lemma, which complements these three lemmas by considering the fourth possible offer type.
	 %
%
	\begin{lemma}
	\label{lemma:lower_j_loss}
	In round $t$, if $p^i$ offers $o \notin L^j_t$ and $p^j$ rejects it, then $L_{t+1}^i \leftarrow L_{t}^j \setminus \{ o' \}$, where $o \neq o'$.
	\end{lemma}
	\begin{proof}
		According to Lemma \ref{lemma:p_2_loss_outcome}, the set $L^i_{t+1}$ contains one outcome less than the set $L^j_t$. Therefore, if $p^i$ offers $o \notin L^j_t$ and $p^j$ rejects it, there must be another outcome $o' \in O_t$ that left the set $L^j_t$.
	\end{proof}

	For ease of notation, we write $U \succ_{p} o$ for $U \subset O$ to denote that party $p$ strictly prefers all of the outcomes in the set $U$ over $o$. 
	The intuition of our proof is as follows. We show that if $p^{max}$ deviates from the strategy specified by Strategy~\ref{str:maxmin}, $p^{min}$ is able to make the negotiation result in an outcome $o$, such that $U^{max}_1 \succ_{p^{max}} o$. 	
	\begin{theorem}
		Strategy~\ref{str:maxmin} specifies a maxmin strategy, and the maxmin value of the game is the set $U^{max}_1$. 
		\label{thm:min_max}
	\end{theorem}
	\begin{proof}
		We will prove by induction on $m$. If $m=2$ WLOG assume that $\succ_{p^{max}}= o_1 \succ o_2$.
		If $p^{max} = p^1$ then $U^{max}_1 = \{o_1,o_2\}$ and clearly one of them will be the negotiation result.
		If $p^{max} = p^2$ then $U^{max}_1 = \{o_1\}$. If $p^{min}$ offers $o_1$ in the first round, according to our strategy $p^{max}$ should accept it. If $p^{min}$ offers $o_2$ in the first round, according to our strategy $p^{max}$ should reject it, and offer $o_1$ in the next round. Since this is the last round, $o_1$ will be accepted. In any case, the negotiation result is $o_1$. On the other hand, if $p^{max}$ deviates and rejects the offer of $o_1$ or accepts the offer of $o_2$, then $o_2$ will be the result of the negotiation, but $U^{max}_1 \succ_{p^{max}} o_2$. 
		Now, assume that if there are $m-1$ outcomes in round $t+1$ our strategy specifies a maxmin strategy, and the maxmin value of the game is the set $U^{max}_{t+1}$. We show that our strategy specifies a maxmin strategy, and the maxmin value of the game is the set $U^{max}_t$ when there are $m$ outcomes in round $t$.
		 
		Assume that it is $p^{max}$'s turn to offer. Clearly, if $p^{max}$ deviates and offers an outcome $o$ such that $U^{max}_t \succ_{p^{max}} o$ then $p^{min}$ can accept it, and the negotiation results in $o$. On the other hand, if $p^{max}$  offers any $o \in U^{max}_t$ then $p^{min}$ can either accept or reject it.
		If $p^{min}$ rejects it then there are $m-1$ outcomes in the next round, and according to the induction assumption $p^{max}$ can guarantee the maxmin value of $U^{max}_{t+1}$ by following our strategy. However, according to Lemma~\ref{lemma:lower_i_not_change}, $L_{t+1}^{max} = L_{t}^{max}$ and thus $U^{max}_{t+1} \cup \{o\} = U^{max}_{t}$. Overall, the maxmin value of the game is the set $U^{max}_{t}$. 
		
		Now assume that it is $p^{min}$'s turn to offer, and $p^{min}$ offers $o \in U^{max}_t$. Clearly, if $p^{max}$ accepts then the negotiation result is from $U^{max}_t$. If $p^{max}$ deviates and rejects, then according to induction assumption $p^{max}$ can guarantee the maxmin value of $U^{max}_{t+1}$. However, according to Lemma~\ref{lemma:lower_j_loss}, $L_{t+1}^{max}= L_{t}^{max} \setminus \{o'\}$, and thus $U^{max}_{t+1} = U^{max}_{t} \setminus \{o\} \cup \{o'\}$. That is, $o'$ is a possible result of the negotiation even though $U^{max}_{t} \succ_{p^{max}} o'$.
		Finally, assume that $p^{min}$ offers $o \notin U^{max}_t$. Clearly, if $p^{max}$ deviates and accepts, then the negotiation results in $o$. On the other hand, if $p^{max}$ follows our strategy and rejects, then according to the induction assumption $p^{max}$ can guarantee the maxmin value of $U^{max}_{t+1}$. However, according to Lemma~\ref{lemma:lower_j_not_change}, $L_{t+1}^{max} = L_{t}^{max} \setminus \{o\}$, and thus $U^{max}_{t+1}=U^{max}_t$.
	\end{proof}
	We note that even though a party does not hold any information regarding the preference order of the other party, she can still guarantee that the negotiation result will be from the upper part of her preference order (i.e., $U^k_1$) by following Strategy~\ref{str:maxmin}. This is possible since both parties have some important common knowledge, which is the number of outcomes $m$, as formally captured in Lemma~\ref{lemma:not_in_lowers}.
	
	Now, what will be the negotiation result if neither party knows the preference order of the other party, but both are rational and will thus follow the maxmin strategy? Clearly, the negotiation result will be an outcome $o$ such that $o \in U^1_1 \cap U^2_1$. That is, an outcome from the set $JG_1$ as defined in Definition~\ref{defn:lowers}. We then get an interesting observation: if $I_1=\emptyset$, $JG_1 = \{0_{eq}\}$ according to Corollary~\ref{lemma:no_inter_eq}, thus the negotiation result is the same for both the case of full information and the case of no information.
	
	In addition, we note that a party $p^i$ cannot guarantee that the negotiation result will be from a subset $U \subset U^i_1$, since we proved that this is the maxmin value. However, she can heuristically offer in each round $t$ the best outcome in $U^i_t$, instead of an arbitrarily chosen $o \in U^i_t$. Since $|U^{j}_t| \geq |L^{j}_t|$, if the other party $p^j$ is also rational and plays the maxmin strategy, there are more cases where $p^j$ will accept this offer, and it is thus beneficial for $p^i$ to heuristically offer in each round $t$ the best outcome in $U^i_t$.
	
	The idea of the maxmin strategy is that a party, not knowing the preferences of the other party, makes a worst case assumption about the behavior of that party (i.e., that she does not need to be rational). This assumption may seem too restrictive, and we therefore also consider the robust-optimization equilibrium solution concept from~\cite{aghassi2006robust}, which we adapt to out setting. Intuitively, in this solution concept each party makes a worst case assumption about the preference order of the other party, but each party still assumes that the other party will play rationally and thus her aim is to maximize her utility.
	Formally, given a strategy profile $[s_1,s_2]$ and a preference order $\prec$, let $w_{s_1,\prec,s_2} = \mathcal{F}([s_1(\prec),s_2(\prec')])$, where $\prec'$ is a preference order such that for all $\prec''$, $\mathcal{F}([s_1(\prec),s_2(\prec')]) \preceq_{p^1} \mathcal{F}([s_1(\prec),s_2(\prec'')])$.  In the robust-optimization setting, a strategy for party $k \in \{1,2\}$,  $s_k$, is a best response to $s_{3-k}$ if for all $s_k'$ and for all $\prec$, $w_{s_k,\prec,s_{3-k}} \succeq_{p^k} w_{s_k',\prec,s_{3-k}}$. A strategy profile $[s_1,s_2]$ is a robust-optimization equilibrium if $s_1$ is a best response to $s_2$ and $s_2$ is a best response to $s_1$. We show that in our setting, surprisingly, every pair of maxmin strategies specifies a robust-optimization equilibrium.	
	\begin{theorem}
		If $s_1$ and $s_2$ are maxmin strategies, then $[s_1,s_2]$ is a robust-optimization equilibrium. 
	\end{theorem}
	\begin{proof}
		Given a preference order, $\prec$, let $\prec_{op}$ be the opposite preference order, i.e., if $\prec = o_1 \prec o_2 \prec ... \prec o_m$ then $\prec_{op} = o_1 \succ o_2 \succ ... \succ o_m$. 
		According to Theorem~\ref{thm:min_max}, if $s_k$ is a maxmin strategy then the negotiation result is $o \in U^k_1$. That is, the worst negotiation result for $p^k$ is the least preferred outcome in $U^k_1$, denoted by $o_{wo}$. Since the other party $p^{3-k}$ is also using a maxmin strategy, $o \in U^k_1 \cap U^{3-k}_1$. For every preference order, $\prec$, if the preference order of $p^{3-k} = \prec_{op}$, $U^k_1 \cap U^{3-k}_1 = \{o_{wo}\}$. That is, $w_{s_1,\prec,s_{3-k}}=o_{wo}$. Assume by contradiction that there is another strategy $s_k'$ and a preference order $\prec$, such that $w_{s_k',\prec,s_{3-k}} \succ_{p^k} o_{wo}$. However, if the preference order of $p^{3-k} = \prec_{op}$, $w_{s_k',\prec,s_{3-k}} \notin U^{3-k}_1$, in contradiction to Theorem~\ref{thm:min_max}.
	\end{proof}
	Finally, consider an asymmetric information setting, where there exists one party that has full information about the other party's preference order, while the other party does not have this information.
	Let $p^{info}$ be the party that has the full information, and $p^{null}$ be the other party.
	$p^{null}$ has no information and she will thus act according to the maxmin strategy (Strategy~\ref{str:maxmin}).
	$p^{info}$ would like to take advantage of her knowledge, so the negotiation result will be better for her. However, according to Theorem~\ref{thm:min_max}, the maxmin value of the game is $U^{null}_1$. Therefore, the best strategy for $p^{info}$ is as follows. If $p^{info}$ starts the negotiation, she should offer the best outcome from $U^{null}_1$ according to her preferences, and $p^{null}$ will accept it. If $p^{info}$ starts the negotiation, she will offer an outcome from $U^{null}_1$. If this is the best outcome according to $p^{info}$'s preferences, she should accept it. Otherwise, in the second round $p^{info}$ should offer the best outcome from $U^{null}_1$ according to her preferences, and $p^{null}$ will accept it.
			
	\section{Conclusion}
	We investigated the VAOV negotiation protocol, which is suitable for ordinal preferences over a finite set of outcomes. We introduced strategies that specify a SPE, and improved upon previous results by providing a linear time algorithm that computes a SPE strategy. We provided substantial analysis of our strategies, which showed the equivalence of the SPE result of the protocol in a non-cooperative setting, to the result of the $RC$ rule in a cooperative setting. 
	Finally, we analyzed the no information setting. 
	We believe that our approach is especially suitable for non-cooperative, multi-agent systems, since we provide easy to implement strategies that can be computed only once if both agents follow the SPE strategy on the equilibrium path. Moreover, there is no need for a central authority to guarantee that the negotiation result will be Pareto optimal, if both agents are rational and follow the SPE strategy. 
	For future work, we would like to extend the protocol to a multi-party setting and analyze the resulting SPE. In addition, it is important to find additional implementation of other bargaining rules by negotiation protocols, similar to the implementation that we showed for the $RC$ rule by the SPE of the VAOV protocol.
	%

	\section*{Acknowledgments}
	This work was supported by the Israel Science Foundation, Grant No. 1488/14.

	\bibliographystyle{abbrv}
	\bibliography{sample}

\begin{thebibliography}{10}

\bibitem{aghassi2006robust}
M.~Aghassi and D.~Bertsimas.
\newblock Robust game theory.
\newblock {\em Mathematical Programming}, 107(1-2):231--273, 2006.

\bibitem{ali2012ordinal}
S.~Ali and S.~Ronaldson.
\newblock Ordinal preference elicitation methods in health economics and health
  services research: using discrete choice experiments and ranking methods.
\newblock {\em British medical bulletin}, 103(1):21--44, 2012.

\bibitem{anbarci93noncooperative}
N.~Anbarci.
\newblock Noncooperative foundations of the area monotonic solution.
\newblock {\em The Quarterly Journal of Economics}, 108(1):245--258, 1993.

\bibitem{anbarci2006finite}
N.~Anbarci.
\newblock Finite alternating-move arbitration schemes and the equal area
  solution.
\newblock {\em Theory and decision}, 61(1):21--50, 2006.

\bibitem{boutilier2015optimal}
C.~Boutilier, I.~Caragiannis, S.~Haber, T.~Lu, A.~D. Procaccia, and O.~Sheffet.
\newblock Optimal social choice functions: A utilitarian view.
\newblock {\em Artificial Intelligence}, 227:190--213, 2015.

\bibitem{brams2001fallback}
S.~J. Brams and D.~M. Kilgour.
\newblock Fallback bargaining.
\newblock {\em Group Decision and Negotiation}, 10(4):287--316, 2001.

\bibitem{handbook}
F.~Brandt, V.~Conitzer, U.~Endriss, A.~D. Procaccia, and J.~Lang.
\newblock {\em Handbook of computational social choice}.
\newblock Cambridge University Press, 2016.

\bibitem{conley2012ordinal}
J.~P. Conley and S.~Wilkie.
\newblock The ordinal egalitarian bargaining solution for finite choice sets.
\newblock {\em Social Choice and Welfare}, 38(1):23--42, 2012.

\bibitem{de2014selection}
G.~De~Clippel, K.~Eliaz, and B.~Knight.
\newblock On the selection of arbitrators.
\newblock {\em American Economic Review}, 104(11):3434--58, 2014.

\bibitem{fatima2014principles}
S.~Fatima, S.~Kraus, and M.~Wooldridge.
\newblock {\em Principles of automated negotiation}.
\newblock Cambridge University Press, 2014.

\bibitem{hurwicz1999designing}
L.~Hurwicz and M.~R. Sertel.
\newblock Designing mechanisms, in particular for electoral systems: the
  majoritarian compromise.
\newblock In {\em Contemporary Economic Issues}, pages 69--88. Springer, 1999.

\bibitem{Bargainingoveraf}
{\"O}.~K{\i}br{\i}s and M.~R. Sertel.
\newblock Bargaining over a finite set of alternatives.
\newblock {\em Social Choice and Welfare}, 28:421--437, 2007.

\bibitem{kraus2001strategic}
S.~Kraus.
\newblock {\em Strategic negotiation in multiagent environments}.
\newblock MIT press, 2001.

\bibitem{mariotti1998nash}
M.~Mariotti.
\newblock Nash bargaining theory when the number of alternatives can be finite.
\newblock {\em Social choice and welfare}, 15(3):413--421, 1998.

\bibitem{nagahisa2002axiomatization}
R.-i. Nagahisa and M.~Tanaka.
\newblock An axiomatization of the kalai-smorodinsky solution when the feasible
  sets can be finite.
\newblock {\em Social Choice and Welfare}, 19(4):751--761, 2002.

\bibitem{nash1950bargaining}
J.~F. Nash~Jr.
\newblock The bargaining problem.
\newblock {\em Econometrica}, 18(2):155--162, 1950.

\bibitem{nunez2015bargaining}
M.~Nunez and J.-F. Laslier.
\newblock Bargaining through approval.
\newblock {\em Journal of Mathematical Economics}, 60:63--73, 2015.

\bibitem{osborne90a}
M.~J. Osborne and A.~Rubinstein.
\newblock {\em Bargaining and Markets}.
\newblock Academic Press, 1990.

\bibitem{rubinstein1982perfect}
A.~Rubinstein.
\newblock Perfect equilibrium in a bargaining model.
\newblock {\em Econometrica: Journal of the Econometric Society}, pages
  97--109, 1982.

\bibitem{shapley1969utility}
L.~S. Shapley.
\newblock {\em Utility comparison and the theory of games}.
\newblock Cambridge: Cambridge University Press. Originally published in La
  Decision, 1969.

\bibitem{sprumont1993intermediate}
Y.~Sprumont.
\newblock Intermediate preferences and rawlsian arbitration rules.
\newblock {\em Social Choice and Welfare}, 10(1):1--15, 1993.

\bibitem{thomson1994cooperative}
W.~Thomson.
\newblock Cooperative models of bargaining.
\newblock {\em Handbook of game theory with economic applications},
  2:1237--1284, 1994.

\bibitem{zhang2008ordinal}
D.~Zhang and Y.~Zhang.
\newblock An ordinal bargaining solution with fixed-point property.
\newblock {\em Journal of Artificial Intelligence Research}, 33:433--464, 2008.

\bibitem{zlotkin1996mechanism}
G.~Zlotkin and J.~S. Rosenschein.
\newblock Mechanism design for automated negotiation, and its application to
  task oriented domains.
\newblock {\em Artificial Intelligence}, 86(2):195--244, 1996.

\end{thebibliography}
		
	\section{Appendix}
Even though the uniqueness of the SPE result was proven by Anbarci~\cite{anbarci93noncooperative}, we provide a direct and simpler proof.
\begin{theorem}
	The SPE result is unique.
\end{theorem}
\begin{proof}
	We prove by induction on $m$. If $m=2$, then no matter what $p^1$ offers, the negotiation results with the most preferred outcome of $p^2$, and thus the SPE is unique.
	Now, assume that if there are $m-1$ outcomes in round $t+1$, the SPE is unique. We show that the SPE is unique when there are $m$ outcomes in round $t$.
	$p^i$ is able to offer an outcome $o \in O_t$. For any such $o$, $p^j$ either accepts $o$ or rejects it and the game moves to round $t+1$ with $m-1$ outcomes. According to the induction assumption, the SPE is unique in each sub-tree of the game where there are $m-1$ outcomes. Since $p^i$ has strict preferences, in a SPE she will choose either an outcome that $p^j$ will accept or a sub-tree of the game, that results with the best outcome according to $p^i$'s preferences. That is, in all of the offers of $p^i$ in round $t$ that are in SPE, the SPE result is the same.
\end{proof}
\end{document}